\documentclass [12pt]{article}
\usepackage{amssymb,amsmath,amsthm}
\usepackage{graphicx}
\usepackage[T1]{fontenc}

\usepackage[cp1250]{inputenc}
\usepackage[active]{srcltx}

\long\def\comment#1{{}}

\newtheorem{thm}{Theorem}[section]
\newtheorem{lemma}[thm]{Lemma}
\newtheorem{cor}[thm]{Corollary}
\newtheorem{rem}[thm]{Remark}

\theoremstyle{definition}

\newtheorem{ex}[thm]{Example}

\begin{document}

\title{{\bf Pairwise Comparisons Simplified }}

\author{W.W. Koczkodaj
\thanks{
Computer Science, Laurentian University,\;\;
Sudbury, Ontario P3E 2C6, Canada 
wkoczkodaj@cs.laurentian.ca 
} \\
\and
J. Szybowski
\thanks{
AGH University of Science and Technology,\; 
Faculty of Applied Mathematics,
al. Mickiewicza 30, 30-059 Krak\'ow, Poland, szybowsk@agh.edu.pl}
\thanks{State Higher Vocational School in Nowy S\k{a}cz,\;\;
Staszica 1, 33-300, Nowy S\k{a}cz, Poland
}\thanks{Research supported by the 
Polish Ministry of Science and Higher Education}
} 

\maketitle

\begin{abstract}

This study examines the notion of generators of a pairwise comparisons matrix. Such approach decreases the number of pairwise comparisons 
from $n\cdot (n-1)$ to $n-1$. An algorithm of reconstructing of the PC matrix from its set of generators is presented. \\


\end{abstract}

\section{Introduction}





In  \cite{Thur27}, Thurstone proposed ``The Law of Comparative Judgments'' for pairwise comparisons (for short, PC). However, the first use of  pairwise comparisons is in \cite{Fech1860}). Even earlier, Condorcet used it in \cite{Condorcet1785}, but in a more simplified way for voting (win or loss). 

The approach, based on the additive value model is presented in \cite{BVCV2012}. Generalization of the pairwise comparisons method was examined in \cite{KO1997,LY2004}.
In \cite{FLR2005}, the PC matrix approximation problem was examined.

In this study, we examine the possibility of reconstructing the entire $n \times n$ PC matrix from only $n-1$ given entries placed in strategic locations. We call them generators. Before we progress, some terminologies of pairwise comparisons  must be revisited in the next section, since PC theory is still not as popular as other mathematical theories. However, the next section is definitely not for PC method experts.

\section{Pairwise comparisons basics}
\label{sec:bc}

We define an $n \times n$ {\em pairwise comparison matrix} simply as a square matrix $M=[m_{ij}]$ such that $m_{ij}>0$ for every 
$i,j=1, \ldots ,n$. A pairwise comparison matrix $M$ is called {\em 
reciprocal} if $m_{ij} = \frac{1}{m_{ji}}$ for every $i,j=1, \ldots ,n$
(then automatically $m_{ii}=1$ for every $i=1, \ldots ,n $). Let
us assume that:

\begin{displaymath}
M = \begin{bmatrix}
1 & m_{12} & \cdots & m_{1n} \\ 
\frac{1}{m_{12}} & 1 & \cdots & m_{2n} \\ 
\vdots & \vdots & \vdots & \vdots \\ 
\frac{1}{m_{1n}} & \frac{1}{m_{2n}} & \cdots & 1
\end{bmatrix}
\end{displaymath}

\noindent where $m_{ij}$ expresses a relative quantity, intensity, or preference of entity (or stimuli) $E_i$ over $E_j$. A more compact and elegant specification of PC matrix is given in \cite{KK2013} by Kulakowski.

A pairwise comparison matrix $M$ is called {\em consistent} (or {\em transitive}) if: 
$$m_{ij} \cdot m_{jk}=m_{ik}$$ 
for every $i,j,k=1,2, \ldots ,n$. \\

We will refer to it as a ``consistency condition''. Consistent PC matrices correspond to the situation with the exact values $\mu(E_1), \ldots , \mu(E_n)$ 
for all the entities. 
In such case, the quotients $m_{ij}=\mu(E_i)/\mu(E_j)$ then form a consistent PC matrix. 
The vector $s=[\mu(E_1), \ldots \mu(E_n)]$ is unique up to a multiplicative constant. 
While every consistent matrix is reciprocal, the converse is generally false. If the consistency condition does not hold, the matrix is {\em inconsistent} (or {\em intransitive}).
Axiomatization of inconsistency indicators for pairwise comparisons has been recently proposed in \cite{KS2013a}. 

The challenge for the pairwise comparisons method comes from the lack of consistency of the pairwise comparisons matrices, which arises in practice (while as a rule, all the pairwise comparisons matrices are reciprocal). Given an $N \times N$ matrix $M$, which is not consistent, the theory attempts to provide a consistent $n \times n$ matrix $M'$, which differs from matrix $M$ ``as little as possible''. 

It is worth to noting that the matrix: $M=[ v_i/v_j]$ is consistent for all (even random) positive values $v_i$. It is an important observation since it implies that a problem of approximation is really a problem of a norm selection and the distance minimization. For the Euclidean norm, the vector of geometric means (equal to the principal eigenvector for the transitive matrix) is the one which generates it. Needless to say that only optimization methods can approximate the given matrix for the assumed norm (e.g., LSM for the Euclidean distance, as recently proposed in \cite{AZG2012}). Such type of matrices are examined in \cite{JT2011} as ``error-free'' matrices.

It is unfortunate that the singular form ``comparison'' is sometimes used considering that a minimum of three comparisons are needed for the method to have a practical meaning. Comparing two entities (stimuli or properties) in pairs is irreducible, since having one entity compared with itself gives trivially 1. Comparing only two entities ($2 \times 2$ PC matrix) does not involve inconsistency. Entities and/or their properties are often called stimuli in the PC research but are rarely used in applications.

\section{The generators of pairwise comparisons matrix }

For a given PC matrix $A \in M_{n\times n}(\mathbb R)$ consider the set $C_n:=\{a_{ij}:\ i<j\}$. Note that to reconstruct the whole matrix it is enough to know the elements of $C_n$, as $a_{ii}=1$ for each $i \in \{1,\cdots,n\}$ and $a_{ji}=\frac{1}{a_{ij}}$ for $i<j$.

Let us call each such set sufficient to reconstruct the matrix $A$ its {\em set of generators}.

The set $C_n$ has $\frac{n^2-n}{2}$ elements. However, consistency is a much stronger condition. So, it is obvious that we can reduce this input set to calculate the rest of elements. It is a natural question to ask which minimal subsets of $C_n$ generate $A$.

\begin{rem}
If $B\subset B' \subset C_n$ and $B$ generates $A$, then $B'$ does as well.
\end{rem}

\begin{thm}\label{th1}
There is no $(n-2)$-set of generators of $A$.
\end{thm}

\begin{proof}
For $n=3$ the statement is obvious, as in any matrix: \\
\[ \left[
\begin{array}{rrr}
1 & a & c \\
\frac{1}{a} & 1 & b \\
\frac{1}{c} & \frac{1}{b} & 1
\end{array}
\right].\]
if we only know one of the values $a$, $b$ or $c$, we cannot clearly calculate the other two satisfying $c=ab$.

To continue the induction, let us assume that the assertion holds for each matrix $M \in M_{n\times n}(\mathbb R)$. Now consider the matrix:

\begin{displaymath}
A_{n+1} = \begin{bmatrix}
1 & a_{12} & \cdots & a_{1n} & a_{1,n+1} \\ 
\frac{1}{a_{12}} & 1 & \cdots & a_{2n} & a_{2,n+1} \\ 
\vdots & \vdots & \vdots & \vdots \\ 
\frac{1}{a_{1n}} & \frac{1}{a_{2n}} & \cdots & 1 & a_{n,n+1}\\
\frac{1}{a_{1,n+1}} & \frac{1}{a_{2,n+1}} & \cdots & \frac{1}{a_{n,n+1}} & 1 
\end{bmatrix}
\end{displaymath}

Notice that in order to calculate the elements of the last column, we need to know at least one of them. On the other hand, if we know $p$ of them, we can calculate only $p-1$ of the elements $a_{ij}$ for $1 \leq i<j\leq n$. Let us assume there is a $(n-1)$-set $B$ of generators of $A_{n+1}$. We define a new set $B':=B \cup L \setminus R$, where $R$ denotes the set of elements of $B$ from the last column, and $L$ denotes the elements from the previous columns which can be calculated from the elements of $R$. Now $B'$ is a $(n-2)$-set of generators of the matrix $A_n$ resulting from $A_{n+1}$ by removing the last row and column. 

$A_n \in M_{n\times n}(\mathbb R)$, which contradicts the inductive assumption.

\end{proof}

\begin{rem} \label{rem-aij}
Given input values $a_{i,i+1}$ (for $i=1,\cdots,n-1$) located above the main diagonal, from the consistency condition we can reconstruct the entire matrix 

\begin{displaymath}
A = \begin{bmatrix}
1 & a_{12} & \cdots & a_{1n} \\ 
\frac{1}{a_{12}} & 1 & \cdots & a_{2n} \\ 
\vdots & \vdots & \vdots & \vdots \\ 
\frac{1}{a_{1n}} & \frac{1}{a_{2n}} & \cdots & 1
\end{bmatrix}
\end{displaymath}
using the formula
\begin{equation} \label{aij}
a_{ij}=\prod_{k=i}^{j-1} a_{k,k+1}
\end{equation} 
for $j > i$.
\end{rem}

It is worth mentioning that the above $n-1$ values are not the only values generating the entire matrix. 
The following theory is provided for finding the other minimal sets of generators.

\begin{rem}\label{mut-rel} There is a mutual relevance between each set $B \subset C_n$ and an undirected graph $G_B$ with $n$ vertices:

\begin{center}
$a_{ij} \in B \Leftrightarrow$ there is an edge $i-j$ in $G_B$. 
\end{center}
\end{rem}

\begin{lemma}\label{lem}
If we know $k$ edges of any subtree of $G_B$, then we are able to compute $\left(\begin{array}{c}k+1\\ 2\end{array}\right)$ of elements of $C_n$.
\end{lemma}

\begin{proof}
For $k=1$ the statement is obvious, as $\left(\begin{array}{c}1+1\\ 2\end{array}\right)=1$.

To proceed with induction, let us assume that the statement is true for $k$. Take a tree $D$ with $k+2$ vertices $x_1,\ldots,x_{k+2} \in V(D)$ and $k+1$ edges. Remove any leaf $x_l \in V(D)$ together with the edge $x_m-x_l$, joining the leaf with the tree. We get a new tree $D'$ with $k$ edges. From the inductive assumption, we are able to calculate $\left(\begin{array}{c}k+1\\ 2\end{array}\right)$ of elements of $C_n$. 

Now, when we give the removed edge back, we notice that for every vertex $x_j \in V(D)\setminus \{l\}$ there is a path $x_j-x_{p_1}- \dots - x_{p_s}-x_l$ joining $x_j$ with $x_l$ and we can compute $a_{jl}=a_{jp_1}\cdot a_{p_1p_2}\cdot \ldots \cdot a_{p_sl}$. These are $k+1$ new elements and altogether we know $\left(\begin{array}{c}k+1\\ 2\end{array}\right)+k+1=\left(\begin{array}{c}k+2\\ 2\end{array}\right)$ of elements of $C_n$. 

\end{proof}

Now we can formulate the necessary and sufficient condition for minimal sets of generators.

\begin{thm} \label{gen-tree} Let us assume $B \subset C_n$ is a $(n-1)$-set. Then
\begin{center}
$B$ generates matrix $A$ $\Leftrightarrow$ $G(B)$ is a tree.
\end{center}
\end{thm}

\begin{proof}
Take a set $B$ of generators of matrix $A$. Let us assume $G(B)$ is not a tree. If so, it must contain a cycle $x_{k_1}-x_{k_2}-\ldots - x_{k_s}- x_{k_1}$. When we remove the edge $x_{k_1}-x_{k_s}$ the relevant $n-2$-set still generates $A$, as $a_{k_1,k_s}=a_{k_1,k_2}\cdot \ldots \cdot a_{k_{s-1},k_s}$. Thus we get a contradiction with Theorem \ref{th1}.

The reverse implication follows straight from Lemma \ref{lem} for $k=n-1$.
\end{proof}

\begin{cor} \label{min-gen}
There are $n^{n-2}$ minimal sets of generators of $A$.
\end{cor}

\begin{proof}
This is an immediate consequence of the Cayley's formula for the number of trees on $n$ vertices.
\end{proof}

Although there are many combinations of $n-1$ values generating the entire PC matrix, the values $m_{i,i+1}$ are the most important of all combinations since they express this sequence:

\begin{equation}
E_1/E_2, E_2/E_3,..., E_{n-2}/E_{n-1}, E_{n-1}/E_n \label{PG1}
\end{equation}

Let us call the above $n-1$ values as {\em principal generators} (PGs). For matrix $A_n$, the principal generators are located above the main diagonal and are as follows:

$$a_{1,2},a_{2,3},\ldots,a_{n-1,n}$$

Let us invent our own handicapping. Handicapping, in sports and games, is the practice of assigning advantage through scoring compensation or other advantage given to different contestants to equalize the chances of winning. The same term also applies to the various methods for computing advantages.

Entities 1 and $n$ occurs in (\ref{PG1}) only once. Any other entity, 2 to $n-1$, occurs twice.
Since the highest frequency is 2, it is fair to add 1 to $E_1$ and $E_n$ to compensate for only one occurrence
and count the maximum number of occurrences of an entity to compensate other entities.
We also define the total handicapping as the total of all compensations.

By the {\em frequency} $f(i,B)$ of an entity $i\in \{1,\ldots,n\}$in a set $B\subset C_n$ we understand the cardinality of set $$B_i:=\{a_{jk}\in B:\ j=i \vee k=i\}.$$

Let $O(B)$ denote the set $\{i:\ B_i\neq \emptyset\}$. We define the {\em total handicapping} of the set $B\subset C_n$ as
$$h(B):=\sum_{i\in O(B)}(\max_{j\in O(B)}f(j,B)-f(i,B)).$$

For example, for generators assumed to be the first raw, the total handicapping would be $(n-2)(n-1)$, since $E_1$ occurs $n-1$ times and the rest entities only once, so they need to be handicapped by $n-2$.

\begin{rem}
The frequency of an entity $i$ in the set $B$ is equal to the degree of vertex $i$ in graph $G_B$:
$$f(i,B)=deg_{G_B}(i).$$
\end{rem}

\begin{rem}
$O(B)$ is the set of vertices of $G_B$ with a positive degree.
\end{rem}

\begin{rem}
$h(B)$ counts the sum of differences between the maximal degree of a vertex in $G_B$ and the degrees of the rest of vertices.
\end{rem}

\begin{cor}
$h(B)=0 \Rightarrow$ graph $G_B$ is regular (all its vertices have the same degree).
\end{cor}

\begin{thm} \label{handi}
If $n>1$ and $B$ is a $(n-1)$-set of generators of $A$, then 
\begin{enumerate}
\item $h(B)=0 \Rightarrow n=2$
\item $h(B)\neq 1$
\item $h(B)=2 \Rightarrow G_B$ is a path connecting all vertices. 
\end{enumerate}
\end{thm}

\begin{proof}
From Theorem \ref{gen-tree}, it follows that $G(B)$ is a tree. Thus, it has $n$ vertices and $n-1$ edges. Notice that to obtain the number of edges one needs to sum up the degrees of vertices and divide by two (each edge is counted twice).

If $h(B)=0$, then $G(B)$ has $n$ vertices of the same degree $k$. Hence,
$$n-1=\frac{nk}{2} \Rightarrow 2n-2=nk \Rightarrow n(2-k)=2 \Rightarrow n=2.$$ 

If $h(B)=1$, then $G(B)$ has $n-1$ vertices of the same degree $k$ and one of degree $k-1$. Hence,
\begin{eqnarray*}n-1=\frac{(n-1)k}{2}+\frac{k+1}{2}=\frac{nk-1}{2} \Rightarrow 2n-2=nk-1  \Rightarrow \\ \Rightarrow (2-k)n=1  \Rightarrow n=1,\end{eqnarray*} 
which is a contradiction.

If $h(B)=2$, then there are two cases. The first one with $n-1$ vertices of degree $k$ and one vertex of degree $k-2$. Hence,
$$n-1=\frac{(n-1)k}{2}+\frac{k-2}{2}=\frac{nk-2}{2} \Rightarrow 2n-2=nk-2  \Rightarrow k=2,$$ and it follows that there is one vertex of degree $0$, so graph $G_B$ is disconnected, which is a contradiction.

For a case with $n-2$ vertices of degree $k$ and two vertices of degree $k-1$, we have:
$$n-1=\frac{(n-2)k}{2}+\frac2{k-1}{2}=\frac{nk-2}{2} \Rightarrow 2n-2=nk-2  \Rightarrow k=2,$$ and it follows that $G_B$ is a path connecting all vertices.

\end{proof}

\begin{cor} \label{min-handi}
For $n \geq 2$ there are $\frac{n!}{2}$ $(n-1)$-sets of generators of $A$ with the minimal total handicapping.
\end{cor}

\begin{proof}
For $n=2$ the statement is obvious. For $n \geq 3$ from Theorem \ref{handi}, we know that it suffices to count all the paths connecting all vertices. Each such path is of the form $p(1)-p(2)-\cdots-p(n)$, where $p$ is a permutation of the set $\{1,\ldots,n\}$. However, from $p(1)-p(2)-\cdots-p(n)$ and $p(n)-p(n-1)-\cdots-p(1)$, we obtain the same tree hence the number of permutations must be divided by two.
\end{proof}

\begin{ex} \label{ex1}
Consider $n=4$ and the matrix
\[ A=\left[
\begin{array}{rrrr}
1 & a & b & c \\
\frac{1}{a} & 1 & d & e \\
\frac{1}{b} & \frac{1}{d} & 1 & f\\
\frac{1}{c} & \frac{1}{e} & \frac{1}{f} & 1 
\end{array}
\right].\]

\noindent The set $C_4=\{a,b,c,d,e,f\}$ consists of $\frac{4^2-4}{2}=6$ elements, so it has \\
 $\left(\begin{array}{c}6\\ 3\end{array}\right)=20$ $3$-subsets listed in Tab. \ref{3sets}. From Remark \ref{mut-rel}, we know that they are related to some graphs. They are listed in Fig.~\ref{fig:shape20} and their nodes are labeled in Fig.~\ref{fig:nodes}.

\begin{table}[h] \label{3sets}
\centering 
\begin{tabular}{|c|c|c|c|}
\hline $\{a,d,f\}$ & $\{a,c,d\}$ & $\{a,c,f\}$ & $\{c,d,f\}$ \\ 
\hline $\{a,e,f\}$ &  $\{a,b,f\}$
& $\{b,c,d\}$ & $\{c,d,e\}$  \\ 
\hline  $\{a,b,e\}$ & $\{b,c,e\}$ & $\{b,e,f\}$ & $\{b,d,e\}$ \\ 
\hline  $\{a,b,c\}$ & $\{a,d,e\}$ & $\{b,d,f\}$ & $\{c,e,f\}$  \\ 
\hline  $\{a,c,e\}$ & $\{a,b,d\}$ & $\{b,c,f\}$ & $\{d,e,f\}$  \\ 
\hline 
\end{tabular}

\caption{$3$-subsets for $n=4$}
\end{table}






\noindent The last four subsets are triads, so they are not generators of $A$. The graphs related to them are cycles.\\

\begin{figure}[h]
\centering
\includegraphics[width=0.5\linewidth]{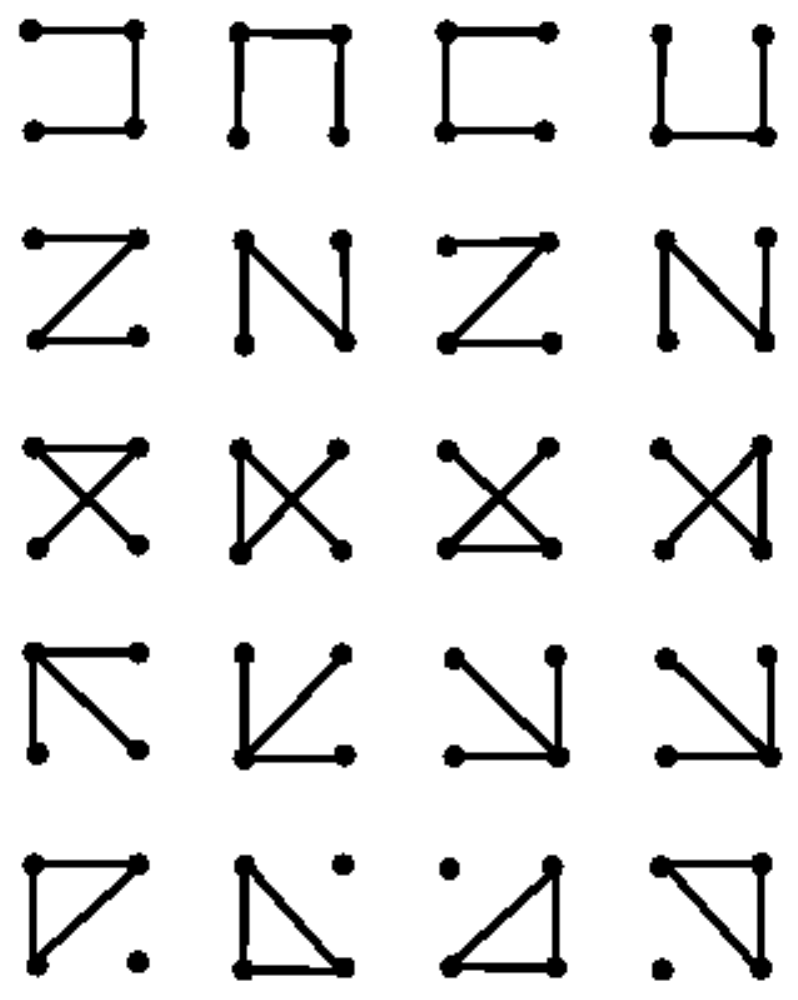}
\caption{All possible combinations of 20 three generators}
\label{fig:shape20}
\end{figure}

\begin{figure}[h]
\centering
\includegraphics[width=0.2\linewidth]{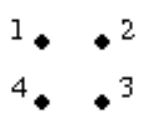}
\caption{Nodes for all generators in Fig.~\ref{fig:shape20}}
\label{fig:nodes}
\end{figure}

According to Corollary \ref{min-gen}, there are $4^{4-2}=16$ minimal sets of generators of $A$ and they are the first $16$. The graphs related to them are obviously trees. From Corollary \ref{min-handi}, we conclude that we have $\frac{4!}{2}=12$ sets of generators of $A$ minimizing the total handicapping and they are the first $12$. The related graphs are paths.

\noindent It is worth noticing that the first subset is the set of PGs.
\end{ex}

The strategic importance of PGs, based on the total handicapping, places them on the diagonal above the main diagonal. It is important to notice that PGs are always there and it is up to us to decide whether or not we use them for the reconstructing the entire consistent PC matrix (with the possible increased error) or enter the remaining entires into PC matrix with the possibility of inconsistency in triads. As pointed out in \cite{KO1998}, nothing is carved in stone and even reciprocity can be relaxed for the blind wine and other types of tasting. Reciprocity may also not be guaranteed when data are collected over the Internet from different sources or individuals.
For PGs sufficiently accurate, the PC matrix reconstruction is a vital possibility. 

\section{An algorithm for the reconstruction of a PC matrix from generators}

Assuming that we know the $(n-1)-$set $B \subset C_n \subset A$, we can reconstruct the entire matrix $A$ using an algorithms described below.

For given PGs $a_{11},\ldots,a_{nn}$, we can compute each $a_{ij}$ for $i<j$  from Remark \ref{rem-aij}. 
On the other hand, if we know a $(n-1)$-set $B$ of different values of matrix $A$ we can apply the $\log$ function to the formula (\ref{aij}) and we substitute $x_k:=\log a_{k,k+1}$ then we get the system of linear equations

$$\log a_{ij}=\sum_{k=i}^{j-1} x_k,\;\;\;\; a_{ij}\in B$$

If $B$ generates $A$, then this system has a unique solution and we may calculate $a_{k,k+1}=10^{x_k}$. \\

\noindent This leads us to the following algorithm: \\

\noindent \underline{\textit{INPUT}}: \\
\begin{itemize}
\item 
$M\in {\mathcal M}(n-1,3)$ representing the set $B$ of generators of $A$. 
\item 
$(M[i,1],M[i,2])$ correspond to  the coordinates of the $i$-th element of $B$.
\item 
$M[i,3]$ corresponds to the value of the $i$-th element of $B$.

\end{itemize}
\noindent \underline{\textit{OUTPUT}}: \\

\noindent The consistent PC matrix $A\in {\mathcal M}(n,n)$ reconstructed from $A$. \\


\noindent \textbf{ALGORITHM}:

\begin{enumerate}
\item Construct graph $G_B$:\\
$G_B:=\emptyset$;\\
for $i:=1$ to $n-1$ $G_B:=G_B \cup \{M[i,1]-M[i,2]\}$;
\item Check if $G_B$ is a tree (using DFS algorithm):\\
if not then write ('$B$ does not generate $A$') and exit;
\item Solve a linear system\\
$$\sum_{k = M[i,1]}^{M[i,2]-1}x_k=\log(M[i,3]),\;\; i=1,\ldots,n-1$$
\item Calculate elements of $A$:\\
for $k:=1$ to $n$ $A[k,k]:=1$;\\
for $k:=1$ to $n-1$ $A[k,k+1]:=10^{x_k}$;\\
for $k:=1$ to $n-1$\\
$\ $\hspace{.5cm} for $l:=k+2$ to $n$ $A[k,l]:=\prod_{m=k}^{l-1}A[m,m+1]$\\
for $k:=1$ to $n-1$\\
$\ $\hspace{.5cm} for $l:=k+1$ to $n$ $A[l,k]:=\frac{1}{A[k,l]}$\\
\end{enumerate}

\begin{ex}

Let us consider the PC matrix $A$, from the example \ref{ex1}, and assume that we know the values from the set $B=\{a,b,e\}$. Thus, the input matrix is:

\[ M=\left[
\begin{array}{rrr}
1 & 2 & a  \\
1 & 3 & b \\
2 & 4 & e
\end{array}
\right].\]

Apply the algorithm:

\begin{enumerate}
\item We obtain the graph $G_B$ with $4$ vertices $$1,\ 2,\ 3,\ 4$$ and $3$ edges $$1-2,\ 1-3,\ 2-4.$$
\item We assure that $G_B$ is a tree and continue.
\item We solve the linear system\\
\[
\left\{
\begin{array}{rcl}
x_1 & = & \log{a}\\
x_1+x_2 & = & \log{b}\\
x_2+x_3 & = & \log{e}
\end{array}
\right.
\]
and we get
\[
\left\{
\begin{array}{rcl}
x_1 & = & \log{a}\\
x_2 & = & \log{b}-\log{a}=\log{\frac{b}{a}}\\
x_3 & = & \log{e}-log{b}+\log{a}=\log{\frac{ae}{b}}
\end{array}
\right.
\]

\item We calculate the missing elements of the output PC matrix $A$:
\[ A=\left[
\begin{array}{rrrr}
1 & a & b & ae \\
\frac{1}{a} & 1 & \frac{b}{a} & e \\
\frac{1}{b} & \frac{a}{b} & 1 & \frac{ae}{b}\\
\frac{1}{ae} & \frac{1}{e} & \frac{b}{ae} & 1 
\end{array}
\right].\]
\end{enumerate}

\end{ex}

\begin{rem}
The complexity of the algorithm is $O(n^3)$.
\end{rem}
\begin{proof}
Steps $1$ and $2$ take $O(n)$ operations. The complexity of step $3$ is $O(n^2)$ and of step $4$ is $O(n^3)$.
\end{proof}

Notice that the assumption that the number $p$ of input entries is equal to $n-1$ may not be satisfied. We may have less or more data and still use the slightly modified algorithm. We only change steps {\textbf{1}} and {\textbf{2}}:

\begin{enumerate}
\item Construct graph $G_B$:\\
$G_B:=\emptyset$;\\
for $i:=1$ to $p$ $G_B:=G_B \cup \{M[i,1]-M[i,2]\}$;
\item Replace $G_B$ by its spanning tree received from DFS algorithm.\\
if $G_B$ has less than $n-1$ edges then write ('$B$ does not generate $A$') and exit;
\end{enumerate}

\section{Conclusions}

In this study, we have demonstrated that we are able to use less than all $n\cdot (n-1)/2$ pairwise comparisons to reconstruct a consistent matrix. In fact, the minimum of $n-1$ generators is sufficient. This may be very useful, provided that they represent more accurate values than the rest of the PC matrix. Another possible application of the suggested algorithm is a case of the missing data in the PC matrix. The entire PC can be reconstructed from the principal generators located above the main diagonal. Together with the preliminary sorting of entities, it is quick and easy way of getting weights.

The problem is that the accumulated errors grow fast when we ``walk away'' from the main diagonal towards the upper right corner. In fact, for $n=7$ and 20\% error in the principal generators, nearly 200\% is accumulated in the upper right corner.

The proposed algorithm for reconstructing the entire matrix from the principal generators is easy to implement, even in Gnumeric (or in MS Excel). In addition, the reconstructed values (accurate or not) are consistent. 

In this study, we only consider the multiplicative variant of PC, which is based on \textit{\textit{``how many times?''}}, while the additive version of pairwise comparisons 
(``by how much?'')was recently analyzed in \cite{KKFY2012}. The additive PC method utilizes a different type of inconsistency (not addressed here). Certainly more research is needed.


\begin{thebibliography}{99}


\bibitem{BVCV2012}
Costa, CABE, De Corte, JM, Vansnick, J-C., MACBETH,International Journal of Information Technology \& Decision Making  11(2): 359-387, 2012.



\bibitem{BCM2013}
Brunelli, M., Canal, L., Fedrizzi, M.,
Inconsistency indices for pairwise comparison matrices: a numerical study, Annals of Operations Research (to appear), 2013.

\bibitem{Condorcet1785}
de Condorcet, N.,``Essay on the Application of Analysis to the Probability of Majority Decisions'', Paris: l'Imprimerie Royale, 1785.

\bibitem{Fech1860}  Fechner, G.T., {\em Elements of Psychophysics}, Vol. 1,
New York: Holt, Rinehart \& Winston, 1965, translation by H.E. Adler of {\em %
Elemente der Psychophysik}, Leipzig: Breitkopf und H\"{a}rtel, 1860.

\bibitem{FLR2005}
Farkas, A, Lancaster, P. Rosza, P., Approximation of positive matrices by transitive matrices, Generalization of the RCGM and LSLR pairwise comparison methods,
Computers \& Mathematics with Applications, 50(7): 1033-1039,2005 .

\bibitem{FKS2010}
Fulop, J., Koczkodaj, W.W., Szarek, S.J., A Different Perspective on a Scale for Pairwise Comparisons,
 Transactions on Computational Collective Intelligence 
 in  Lecture Notes in Computer Science, 6220, 71-84, 2010.


\bibitem{AZG2012}
Grzybowski, AZ, 
Note on a new optimization based approach for estimating priority weights and related consistency index,
Expert Systems with Applications, 39(14): 11699-11708, 2012.

\bibitem{HerKocz96}  Herman, M., Koczkodaj, W.W., 
Monte Carlo Study of Pairwise Comparisons, Information Processing Letters, 57(1), pp. 25-29,
1996.




\bibitem{Kocz93}  
Koczkodaj, W.W., 
A New Definition of Consistency of Pairwise Comparisons. 
Mathematical and Computer Modelling, 18(7), 79-84, 1993.

\bibitem{K1998}
Koczkodaj, W.W., Testing the Accuracy Enhancement of Pairwise Comparisons by a Monte Carlo Experiment, Journal of Statistical Planning and Inference, 69(1), pp. 21-32, 1998.

\bibitem{KO1997} 
Koczkodaj, W.W., Orlowski, M.,
An orthogonal basis for computing a consistent approximation 
to a pairwise comparisons matrix.,
Matrix, Computers \& Mathematics with Applications 34(10): 41-47, 1997.

\bibitem{KO1998} 
Koczkodaj, W.W., Orlowski, M.,
Computing a Consistent Approximation to a Generalized Pairwise Comparisons Matrix, Computers \& Mathematics with Applications 37(3): 79-85, 1999.

\bibitem{KS2010} Koczkodaj, W.W., Szarek, S.J., On distance-based inconsistency reduction algorithms for pairwise comparisons,
Logic J. of the IGPL, 18(6): 859-869, 2010.

\bibitem{KS2013a} 
Koczkodaj, W.W.,Szwarc, R., 
On Axiomatization of Inconsistency Indicators for Pairwise Comparisons
arXiv.org (accepted for publication in Fundamenta Informaticae)

\bibitem{KK2013}
Kulakowski, K.,
A heuristic rating estimation algorithm for the pairwise comparisons method, European Journal of Operations Research, (online) 2013

\bibitem{LY2004}
Limayem, F, Yannou, B.,
Generalization of the RCGM and LSLR pairwise comparison methods,
Computers \& Mathematics with Applications, 48(3-4): 539-548, 2004.





\bibitem{JT2011} %
Temesi, J.,
Pairwise comparison matrices and the error-free property of the decision maker, Central European Journal of Operations Research, 19(2): 239-249, 2011.

\bibitem{Thur27}  Thurstone, L.L., {\em A Law of Comparative Judgments}, Psychological Reviews, Vol. 34, 273-286, 1927.


\bibitem{KKFY2012}
Yuen, K.K.F., 
Pairwise opposite matrix and its cognitive prioritization operators: comparisons with pairwise reciprocal matrix and analytic prioritization operators,
Journal of the Operational Research Society, 63(3): 322-338, 2012.

\end{thebibliography}
\end{document}